\documentclass[envcountsame]{llncs}

\newcommand{\longversion}[1]{#1}
\newcommand{\shortversion}[1]{}

\usepackage{amsmath, amssymb}
\usepackage{tikz}
\usepackage{graphicx}

\usepackage{mathrsfs}
\usepackage{enumitem}

\newcommand{\Nat}{\mathbb{N}}
\newcommand{\hy}{\hbox{-}\nobreak\hskip0pt}

\def\hy{\hbox{-}\nobreak\hskip0pt}

\newcommand{\SB}{\{\,} \newcommand{\SM}{\;{:}\;} \newcommand{\SE}{\,\}}
 
\newcommand{\Card}[1]{|#1|}

\newcommand{\CCC}{\mathcal{C}}

\newcommand{\NP}{\text{\normalfont NP}}

\newcommand{\FPT}{\text{\normalfont FPT}}

\newcommand{\W}[1][xxxx]{\text{\normalfont W}[#1]}

\newcommand{\mtext}[1]{\text{\normalfont\itshape #1}} 
\newcommand{\var}{\mtext{var}}

\newcommand{\ol}[1]{\overline{#1}}

\newcommand{\proj}[2]{#1|_{#2}}

\newcommand{\shapes}[1]{\mathsf{shapes}(#1)}
\newcommand{\generators}[2]{\mathsf{generators}_{#1}(#2)}
\newcommand{\rgenerators}[2]{\mathsf{restricedgen}_{#1}(#2)}
\newcommand{\rshapes}[1]{\mathsf{rshapes}(#1)}
\newcommand{\projset}[2]{\mathsf{proj}(#1,#2)}
\newcommand{\tuplestop}[1]{\mathscr{X}_{#1}}
\newcommand{\tuplesbot}[1]{\overline{\mathscr{X}_{#1}}}
\newcommand{\infunctions}[1]{\mathsf{infunctions}(#1)}
\newcommand{\outfunctions}[1]{\mathsf{outfunctions}(#1)}
\newcommand{\shapeout}[1]{\mathit{out}_{#1}}
\newcommand{\shapein}[1]{\mathit{in}_{#1}}

\let\doendproof\endproof
\renewcommand\endproof{~\hfill$\qed$\doendproof}

\begin{document}

\title{Model Counting for Formulas of\\Bounded Clique\hy Width\thanks{This research was supported by the
    ERC (COMPLEX REASON, 239962).}}

\author{Friedrich Slivovsky \and Stefan Szeider}

\institute{
 Institute of Information Systems,
 Vienna University of Technology, Vienna, Austria
\email{fs@kr.tuwien.ac.at,stefan@szeider.net}
}

\maketitle

\begin{abstract}\begin{sloppypar}
  We show that \#SAT is polynomial\hy time tractable for classes of
  CNF formulas whose incidence graphs have bounded symmetric clique\hy width
  (or bounded clique\hy width, or bounded rank\hy width). This result
  strictly generalizes polynomial\hy time tractability results for
  classes of formulas with signed incidence graphs of bounded
  clique\hy width and classes of formulas with incidence graphs of
  bounded modular treewidth, which were the most general results of
  this kind known so far.
\end{sloppypar}
\end{abstract}

\section{Introduction}
Propositional model counting (\#SAT) is the problem of computing the number of
satisfying truth assignments for a given CNF
formula
. It is a well-studied problem with
applications in Artificial Intelligence, such as probabilistic
inference~\cite{BacchusDalmaoPitassi03,SangBeameKautz05}. It is also a
notoriously hard problem: \#SAT is \#P-complete in general~\cite{Valiant79b}
and remains \#P-hard even for monotone 2CNF formulas and Horn 2CNF
formulas~\cite{Roth96}. It is \NP-hard to approximate the number of satisfying
truth assignments of a formula with $n$ variables to within $2^{n^{1 -
    \varepsilon}}$ for any $\varepsilon > 0$. As in the exact case, this
hardness result even holds for monotone 2CNF formulas and Horn 2CNF
formulas~\cite{Roth96}.  While these syntactic restrictions do not make the
problem easier, \#SAT becomes tractable under certain \emph{structural}
restrictions~\cite{FischerMakowskyRavve06,GanianHlinenyObdrzalek13,GaspersSzeider13,NishimuraRagdeSzeider07,OrdyniakPaulusmaSzeider13,PaulusmaSlivovskySzeider13,SamerSzeider10,Szeider04b}. Structural
restriction are obtained by bounding parameters of (hyper)graphs associated
with formulas. We extend this line of research and study \#SAT for classes of
formulas whose incidence graphs (that is, the bipartite graph whose vertex
classes consist of variables and clauses, with variables adjacent to clauses
they occur in) have bounded \emph{symmetric clique\hy
  width}~\cite{Courcelle04}. Symmetric clique\hy width is a parameter that is
closely related to clique\hy width, rank\hy width, and Boolean\hy width: a class of graphs has
bounded symmetric clique\hy width \textit{iff} it has bounded clique\hy width
\textit{iff} it has bounded rank\hy width \textit{iff} it has bounded
Boolean\hy width. For a graph class $\CCC$, let
$\#\textnormal{SAT}(\CCC)$ be the restriction of \#SAT to instances $F$ with
incidence graph $I(F) \in \CCC$. We prove:
\begin{theorem}\label{thmmain}\sloppypar
  $\#\textnormal{SAT}(\CCC)$ is polynomial\hy time tractable for any graph
  class $\CCC$ of bounded symmetric clique\hy width.
\end{theorem}
This result generalizes polynomial\hy time tractability results for classes of
formulas with signed incidence graphs of bounded clique\hy
width~\cite{FischerMakowskyRavve06} and classes of formulas with incidence
graphs of bounded modular treewidth~\cite{PaulusmaSlivovskySzeider13}. The
situation is illustrated in Figure~\ref{fighierarchy} (for a survey of results
for width\hy based parameters, see
\cite{OrdyniakPaulusmaSzeider13,PaulusmaSlivovskySzeider13}).
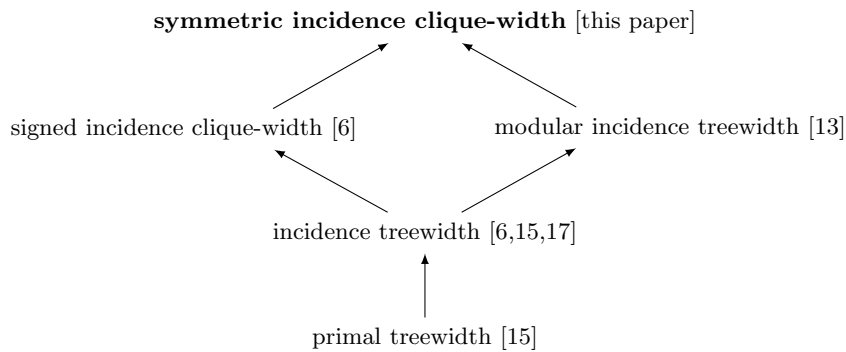
\begin{figure*}
\begin{center}
\small
\begin{tikzpicture}[>=latex,auto,yscale=0.7]
\node (CW) at (0,6) {\textbf{symmetric incidence clique-width} [this paper]};
\node (MTW) at (2.5,4) {~~~~~~~~~~~~~~modular incidence treewidth~\cite{PaulusmaSlivovskySzeider13}};
\node (SCW) at (-2.5,3.97) {signed incidence
    clique-width~\cite{FischerMakowskyRavve06}~~~~~~~~~~~~~~};
\node (ITW) at (0,2) {incidence treewidth~\cite{FischerMakowskyRavve06,SamerSzeider10,Szeider04b}};
\node (PTW) at (0,0) {primal treewidth~\cite{SamerSzeider10}};
\path[thin, ->] (PTW) edge (ITW);
\path[thin, ->] (ITW) edge (SCW);
\path[thin, ->] (ITW) edge (MTW);
\path[thin, ->] (SCW) edge (CW);
\path[thin, ->] (MTW) edge (CW);
\end{tikzpicture}
\caption{A hierarchy of structural parameters. An arc from a parameter
  $p$ to a parameter $q$ reads as ``for any class of formulas, $q$ is bounded whenever $p$ is
  bounded.''  Bold type is used to indicate parameters that render
  \#SAT polynomial-time tractable when bounded by a constant.}
\label{fighierarchy}
\end{center}
\end{figure*}
Our result is obtained through a combination of dynamic
programming on a decomposition tree with the representation of truth
assignments by \emph{projections} (i.e., sets of clauses satisfied by
these assignments). This extends the
techniques used to prove polynomial\hy tractability of \#SAT for
classes of formulas with incidence graphs of bounded modular treewidth
\cite{PaulusmaSlivovskySzeider13}; there, partial assignments are
partitioned into equivalence classes by an equivalence relation
roughly defined as follows: two assignments are equivalent whenever
they satisfy the same set of clauses of a certain formula induced by a
subtree of the decomposition. To make bottom\hy up dynamic programming
work, it is enough to record the number of assignments in each
equivalence class. This approach does not carry over to the case of
bounded symmetric clique\hy width for principal reasons: the number of
equivalence classes of such a relation can be exponential in the size
of the (sub)formula.

To deal with this, our algorithm uses the technique of taking into
account an ``expectation from the outside''
\cite{BuixuanTelleVatshelle10,GanianHlineny10,GanianHlinenyObdrzalek13}.
The underlying idea is that the information one has to record for any
particular partial solution can be reduced significantly if one
includes an ``expectation'' about what this partial solution will be
combined with to form a complete solution. This trick allows us to
bound the number of records required for dynamic programming by a
polynomial in the number of clauses of the input formula.

For all parameters considered in Figure~\ref{fighierarchy},
propositional model counting is polynomial\hy time tractable if the
parameter is bounded by a constant, but some of them even admit so\hy
called {\FPT} algorithms. The runtime of an {\FPT} algorithm is
bounded by a function of the form $f(k)\:p(l)$, where $f$ is an
arbitrary computable function and $p$ is a polynomial with order
independent of the parameter $k$. As we will see, the order of the
polynomial bounding the runtime in Theorem~\ref{thmmain} is dependent
on the parameter. One may wonder whether this can be avoided, that is,
whether the problem admits an {\FPT} algorithm. The following result
shows that this is not possible, subject to an assumption from
parameterized complexity.
\begin{theorem}[\cite{OrdyniakPaulusmaSzeider13}]\label{thmhardness}
{\sc SAT}, parameterized by the symmetric clique\hy width of the incidence graph of the input formula,
is $\W[1]$\hy hard.
\end{theorem}
To be precise, the result proven in \cite{OrdyniakPaulusmaSzeider13}
is stated in terms of clique\hy width. However, since the clique\hy
width of a graph is at most twice its symmetric clique\hy width (see~\cite{Courcelle04}), the result carries over to symmetric clique\hy
width. 
\section{Preliminaries}\label{sec:preliminaries}
Let $f: X \rightarrow Y$ be a function and $X' \subseteq X$. We let
$f(X')= \SB f(x) \in Y \SM x \in X' \SE$. Let $X^*$
and $Y^*$ be sets, and let $g:X^*\rightarrow Y^*$ be a function with
$g(x)=f(x)$ for all $x\in X\cap X^*$. Then the function $f\cup g:
X\cup X^*\to Y\cup Y^*$ is defined as $(f\cup g)(x)=f(x)$ if $x\in X$
and $(f\cup g)(x)=g(x)$ if $x\in X^*\setminus X$.

\paragraph{Graphs.} The graphs considered in this paper are loopless, simple, and undirected. If
$G$ is a graph and $v$ is a vertex of $G$, we let $N(v)$ denote the set of all
neighbors of $v$ in $G$. For a tree $T$ we write $L(T)$ to denote the set of
leaves of~$T$. Let $\CCC$ be a class of graphs and let $f$ be a mapping
(invariant under isomorphisms) that associates each graph $G$ with a non\hy
negative real number. We say~$\CCC$ \emph{has bounded $f$} if there is a $c$
such that $f(G) \leq c$ for every $G \in \CCC$.
\paragraph{Formulas.}
We assume an infinite supply of propositional \emph{variables}. A
\emph{literal} is a variable $x$ or a negated variable $\ol{x}$; we
put $\var(x) = \var(\ol{x}) = x$; if $y=\ol{x}$ is a literal, then we
write $\ol{y}=x$.  For a set $S$ of literals we write
$\ol{S}=\SB\ol{x} \SM x\in S\SE$; $S$ is \emph{tautological} if $S\cap
\ol{S}\neq \emptyset$.  A \emph{clause} is a finite non-tautological
set of literals.  A finite set of clauses is a \emph{CNF formula} (or
\emph{formula}, for short).  The \emph{length} of a formula $F$ is
given by $\sum_{C\in F}|C|$.  A variable $x$ \emph{occurs} in a clause
$C$ if $x \in C\cup \ol{C}$. We let $\var(C)$ denote the set of
variables that occur in $C$.  A variable $x$ \emph{occurs} in a
formula $F$ if it occurs in at least one of its clauses, and we let
$\var(F)=\bigcup_{C\in F} \var(C)$. If $F$ is a formula and $X$ a set
of variables, we let $\proj{F}{X} = \SB C \in F \SM X \subseteq
\var(C) \SE$. The \emph{incidence graph} of a formula $F$ is the
bipartite graph $I(F)$ with vertex set $\var(F)\cup F$ and edge set
$\SB Cx \SM C\in F$ and $x\in \var(C)\SE$.

Let $F$ be a formula. A \emph{truth assignment} is a mapping $\tau: X
\rightarrow \{0, 1\}$ defined on some set of variables $X \subseteq
\var(F)$. We call $\tau$ \emph{total} if $X = \var(F)$ and
\emph{partial} otherwise. For $x\in X$, we define
$\tau(\overline{x})=1-\tau(x)$. A truth assignment $\tau$
\emph{satisfies} a clause $C$ if $C$ contains some literal~$\ell$ with
$\tau(\ell)=1$. If $\tau$ satisfies all clauses of $F$, then $\tau$
\emph{satisfies} $F$; in that case we call $F$ satisfiable. The
\emph{Satisfiability} (SAT) problem is that of testing whether a given
formula is satisfiable. The \emph{propositional model counting}
(\#SAT) problem is a generalization of SAT that asks for the number of
satisfying total truth assignments of a given formula. For a graph
class $\CCC$, we let $\#\textnormal{SAT}(\CCC)$ be the restriction of
\#SAT to instances $F$ with $I(F) \in \CCC$.

\paragraph{Decomposition Trees.}
We review decomposition trees following the presentation
in~\cite{BuixuanTelleVatshelle11}. Let $G = (V,E)$ be a graph. A
\emph{decomposition tree} for $G$ is a pair $(T, \delta)$, where $T$
is a rooted binary tree and $\delta: L(T) \rightarrow V$ is a
bijection. For a subset $X \subseteq V$ let $\overline{X} = V
\setminus X$. We associate every edge $e \in E(T)$ with a bipartition
$P_e$ of~$V$ obtained as follows. If $T_1$ and $T_2$ are the
components obtained by removing~$e$ from $T$, we let $P_e = (L(T_1),
L(T_2))$. Note that $L(T_2) = \overline{X}$ for $X = L(T_1)$. A
function $f: 2^V \rightarrow \mathbb{R}$ is \emph{symmetric} if $f(X)
= f(\overline{X})$ for all $X \subseteq V$. Let $f: 2^V \rightarrow
\mathbb{R}$ be a symmetric function. The $f$\hy width of $(T, \delta)$
is the maximum of $f(X) = f(\overline{X})$ taken over the bipartitions
$P_e = (X, \overline{X})$ for all $e \in E(T)$. The $f$\hy width of
$G$ is the minimum of the $f$\hy widths of the decomposition trees of
$G$.

Let $A(G)$ stand for the \emph{adjacency matrix} of $G$, that is, the $V
\times V$ matrix $A(G) = (a_{vw})_{v \in V, w \in V}$ such that $a_{vw} = 1$
if $vw \in E$ and $a_{vw} = 0$ otherwise. For $X, Y \subseteq V$, let
$A(G)[X,Y]$ denote the $X \times Y$ submatrix $(a_{vw})_{v \in X, w \in
  Y}$. The \emph{cut\hy rank} function $\rho_G: 2^V \rightarrow \mathbb{R}$ of
$G$ is defined as
\begin{align*}
  \rho_G(X)  = \mathit{rank}(A(G)[X, V \setminus X]),
\end{align*}
where $\mathit{rank}$ is the rank function of matrices over $\mathbb{Z}_2$. The row
and column ranks of any matrix are equivalent, so this function is symmetric.
The \emph{rank\hy width} of a decomposition tree $(T, \delta)$ of $G$, denoted
$\mathit{rankw}(T, \delta)$, is the $\rho_G$\hy width of $(T, \delta)$, and
the \emph{rank\hy width} of $G$, denoted $\mathit{rankw}(G)$, is the
$\rho_G$\hy width of $G$.

Let $X$ be a proper nonempty subset of $V$. We define an equivalence relation~$\equiv_X$
on $X$ as
\begin{align*}
  x \equiv_X y \textnormal{ iff, for every $z \in V \setminus X$, } xz
  \in E \Leftrightarrow yz \in E.
\end{align*}
The \emph{index} of $X$ in $G$ is the cardinality of $X/\mathord\equiv_X$, that is,
the number of equivalence classes of $\equiv_X$. We let $\mathit{index}_G: 2^V
\rightarrow \mathbb{R}$ be the function that maps each proper nonempty subset
$X$ of $V$ to its index in $G$. We now define the function $\iota_G: 2^V
\rightarrow \mathbb{R}$ as
\begin{align*}
  \iota_G(X) = \max(\mathit{index}_G(X), \mathit{index}_G(V \setminus X)).
\end{align*}
\begin{sloppypar}
This function is trivially symmetric. The \emph{index} of a decomposition tree $(T,
\delta)$ of~$G$, denoted $\mathit{index}(T, \delta)$, is the $\iota_G$\hy
width of $(T, \delta)$. The \emph{symmetric clique\hy width}~\cite{Courcelle04} of~$G$, denoted
$\mathit{scw}(G)$, is the $\iota_G$\hy width of $G$.
\end{sloppypar}

Symmetric clique\hy width and rank\hy width are closely related graph
parameters. In fact, the index of a decomposition tree can be bounded in terms
of its rank\hy width.
\begin{lemma}\label{lemscwrankw}
  For every graph $G$ and decomposition tree $(T, \delta)$ of $G$, $\mathit{rankw}(T,\delta) \leq \mathit{index}(T, \delta) \leq
  2^{\mathit{rankw}(T, \delta)}$.
\end{lemma}
\newcommand{\pflemmascwrankw}[0]{
\begin{proof}
  Let $G = (V,E)$ be a graph and $X$ be a nonempty proper subset of $V$. For
  every pair of vertices $x,y \in X$ the rows of $A(G)[X, V \setminus X]$ with
  indices $x$ and $y$ are identical if and only if $x \equiv_X y$. So
  $\mathit{index}_G(X)$ is precisely the number of distinct rows of $A(G)[X, V
  \setminus X]$, which is an upper bound on the rank of $A(G)[X, V \setminus
  X]$ over $\mathbb{Z}_2$. Symmetrically, $\mathit{index}_G(V \setminus X)$ is the
  number of distinct columns of $A(G)[X, V \setminus X]$, which is also an
  upper bound on the rank. So $\rho_G(X) \leq \iota_G(X)$, which proves the
  left inequality. The rank of $A(G)[X, V \setminus X]$ is the cardinality of a basis
  for the matrix's row (column) space. That is, each of its row (column)
  vectors can be represented as a linear combination of $\rho_G(X)$ row
  (column) vectors. Over $\mathbb{Z}_2$, any linear combination can be obtained using
  only $0$ and $1$ as coefficients. Accordingly, there can be at most
  $2^{\rho(X)}$ distinct rows (columns) in $A(G)[X, V \setminus X]$. So
  $\iota_G(X) \leq 2^{\rho_G(X)}$, and the right inequality follows.
\end{proof}}
\longversion{\pflemmascwrankw}
\begin{corollary}\label{corscwrankw}
  For every graph $G$, $\mathit{rankw}(G) \leq  \mathit{scw}(G) \leq 2^{\mathit{rankw}(G)}$.
\end{corollary}
Runtime bounds for the dynamic programming algorithm presented below are more
naturally stated in terms the index of the underlying decomposition tree than
in terms of its rank\hy width. However, to the best of our knowledge, there is
no polynomial\hy time algorithm for computing decomposition trees of minimum
index directly -- instead, we will use the following result to compute
decomposition trees of minimum rank\hy width.
\begin{theorem}[\cite{HlinenyOum08}]\label{thmrankdecomp} Let $k \in \Nat$ be a constant and
  $n \geq 2$. For an $n$-vertex graph~$G$, we can output a
  decomposition tree of rank\hy width at most $k$ or confirm that the
  rank-width of $G$ is larger than $k$ in time $O(n^3)$.
\end{theorem}

\paragraph{Projections.}
Let $F$ be a set of clauses and $X$ a set of variables. For an assignment
$\sigma \in 2^{X}$ we write $F(\sigma)$ to denote the set of clauses of $F$
satisfied by $\sigma$, and call~$F(\sigma)$ a \emph{projection} of $F$. We
write $\projset{F}{X} = \SB F(\sigma) \SM \sigma \in 2^X \SE$ for the set of
projections of $F$ with respect to a set $X$ of variables.

\begin{proposition}\label{propprojbounded} 
  Let $F$ be a formula with $m$ clauses and let $X \subseteq \var(F)$ be a set
  of variables. We have $\Card{\projset{\proj{F}{X}}{X}} \leq m +
  1$. Moreover, the set $\projset{\proj{F}{X}}{X}$ can be computed in time
  polynomial in $l$, where $l$ is the length of $F$.
\end{proposition}
\begin{proof}
  Let $\sim_X$ be the relation on clauses defined as $C \sim_X C'$ if
  $\SB \ell \in C \SM \var(\ell)\in X\SE = \SB \ell \in C' \SM
  \var(\ell) \in X \SE$. Clearly $\sim_X$ is an equivalence
  relation. Let $\CCC_1, \dots, \CCC_l$ be the equivalence classes of
  $\sim_X$ on $\proj{F}{X}$. Recall that every clause $C$ in
  $\proj{F}{X}$ contains all variables in $X$. As a consequence, an
  assignment $\tau \in 2^X$ either satisfies all clauses in
  $\proj{F}{X}$ or it satisfies all clauses in $\proj{F}{X}$ except
  those in a unique class $\CCC_i$ for $i \in \{1,\dots,l\}$, in which
  case $\proj{F}{X}(\tau) = \proj{F}{X} \setminus \CCC_i$. Since
  $\proj{F}{X} \subseteq F$ we get $l \leq m$, and thus
  $\Card{\projset{\proj{F}{X}}{X}} \leq m + 1$. Computing
  $\projset{\proj{F}{X}}{X}$ boils down to computing
  $\CCC_1,\dots,\CCC_l$ and in turn $\proj{F}{X} \setminus \CCC_i$ for
  each $i \in \{1,\dots,l\}$, which can be done in time polynomial in
  the length of $F$. The set $\proj{F}{X}$ is contained in
  $\projset{\proj{F}{X}}{X}$ if and only if $l < 2^{\Card{X}}$, which
  can be checked in polynomial time as well.
\end{proof}

\section{An Algorithm for \#SAT}\label{sec:algorithm}
In this section, we will describe an algorithm for \#SAT via dynamic
programming on a decomposition tree. \shortversion{Due to space constraints, several
proofs are placed in the appendix. }To simplify the statements of
intermediate results, we fix a formula $F$ with $\Card{F} = m$ clauses
and a decomposition tree $(T,\delta)$ of $I(F)$ with
$\mathit{index}(T, \delta) = k$. For a node $z \in V(T)$, let $T_z$
denote the maximal subtree of $T$ rooted at $z$.  We write $\var_z$
for the set of variables $\var(F) \cap \delta(L(T_z))$ and $F_z$ for
the set of clauses $F \cap \delta(L(T_z))$. Moreover, we let
$\overline{F_z} = F \setminus F_z$ and $\overline{\var_z} = \var(F)
\setminus \var_z$.

Our algorithm combines techniques from~\cite{PaulusmaSlivovskySzeider13} with
dynamic programming using
``expectations''~\cite{BuixuanTelleVatshelle10,GanianHlineny10,GanianHlinenyObdrzalek13}.
We briefly describe the information maintained for each node $z \in V(T)$ of
the decomposition. Classes of truth assignments $\sigma \in 2^{\var_z}$ will
be represented by two sets of clauses. The first set (typically denoted
$\shapeout{}$) corresponds to the projection $\overline{F_z}(\sigma)$, that
is, the set of clauses \emph{outside} the current subtree that is satisfied by
$\sigma$. The second set is a projection~$F_z(\tau)$ for some $\tau \in
2^{\overline{\var_z}}$ so that the combined assignment $\sigma \cup \tau$
satisfies $F_z$. This set of clauses (typically denoted $\shapein{}$) is
``expected'' to be satisfied from outside the current subtree by an
``incoming'' assignment. Adopting the terminology
of~\cite{GanianHlinenyObdrzalek13}, we call these pairs of sets \emph{shapes}.
\begin{definition}[Shape]\label{defshape}
  Let $z \in V(T)$, let $\shapeout{z} \subseteq \overline{F_z}$, and let
  $\shapein{z} \subseteq F_z$. We call the pair $(\shapeout{z}, \shapein{z})$
  a \emph{shape (for z)}, and say an assignment $\tau \in 2^{\var_z}$ is
  \emph{of shape} $(\shapeout{z}, \shapein{z})$ if it satisfies the following
  conditions.
  \begin{enumerate}[label=(\roman*), leftmargin=2pc, itemsep=1em]
  \item $\overline{F_z}(\tau) = \shapeout{z}$. \label{cond:shape1}
  \item For each clause $C \in F_z$, the assignment $\tau$ satisfies $C$ or $C
    \in \shapein{z}$. \label{cond:shape2}
  \end{enumerate}
  If $\shapeout{z} \in \projset{\overline{F_z}}{\var_z}$ and
  $\shapein{z} \in \projset{F_z}{\overline{\var_z}}$ then the shape
  $(\shapeout{z}, \shapein{z})$ is \emph{proper}. We denote the set of
  shapes for $z \in V(T)$ by $\shapes{z}$ and write $N_z(s)$ to denote
  the set of assignments in $2^{\var_z}$ of shape $s \in
  \shapes{z}$. Moreover, we let $n_z(s)=\Card{N_z(s)}$.
\end{definition}
Note that an assignment can have multiple shapes, so shapes do not
partition assignments into equivalence classes.
\begin{lemma}\label{lememptyproper}
  A truth assignment $\tau \in 2^{\var(F)}$ satisfies $F$ if and only
  if it has shape~$(\emptyset, \emptyset)$. Moreover, the shape $(\emptyset, \emptyset)$ is proper.
\end{lemma}
\newcommand{\pflememptyproper}[0]{
\begin{proof}
  Observe that $\var_r = \var(F)$, and let $\tau \in
  2^{\var_r}$. Suppose $\tau$ satisfies $F$. Since $\overline{F_r}$ is
  empty, we immediately get $\overline{F_r}(\tau) = \emptyset$, so
  $\tau$ satisfies condition~\ref{cond:shape1}. Moreover $\tau$
  satisfies every clause of $F = F_r$, so condition~\ref{cond:shape2}
  is satisfied as well. For the right to left direction, suppose
  $\tau$ has shape $(\emptyset, \emptyset)$. It follows from
  condition~\ref{cond:shape2} that $\tau$ must satisfy $F_r = F$. To
  see that $(\emptyset, \emptyset)$ is proper note that
  $\overline{F_r}(\sigma) = \emptyset$ for any $\sigma \in
  2^{\var_r}$, and that $2^{\overline{\var_r}}$ contains only the
  empty function $\epsilon: \emptyset \rightarrow \{0,1\}$ with
  $F_r(\epsilon) = \emptyset$.
\end{proof}}
\longversion{\pflememptyproper}
This tells us that $n_r((\emptyset, \emptyset))$ is equal
to the number of satisfying truth assignments of $F$. Let $x,y,z \in
V(T)$ such that $x$ and $y$ are the children of $z$, and let $s_x,
s_y,s_z$ be shapes for $x,y,z$, respectively. The assignments in
$N_x(s_x)$ and~$N_y(s_y)$ contribute to $N_z(s_z)$ if certain
conditions are met. These are captured by the following definition.
\begin{definition}\label{defgen}
  Let $x, y, z \in V(T)$ such that $x$ and $y$ are the children of $z$.
  We say two shapes $(\shapeout{x},\shapein{x}) \in \shapes{x}$ and
  $(\shapeout{y}, \shapein{y}) \in \shapes{y}$ \emph{generate} the
  shape $(\shapeout{z}, \shapein{z}) \in \shapes{z}$ whenever the
  following conditions are satisfied.
  \begin{enumerate}[label =(\arabic*), leftmargin=2pc, itemsep=1em]
  \item $\mathit{out}_z = (\mathit{out}_x \cup \mathit{out}_y) \cap \overline{F_z}$ \label{cond:gen1}
  \item $\mathit{in}_x = (\mathit{in}_z \cup \mathit{out}_y) \cap F_x$ \label{cond:gen2}
  \item $\mathit{in}_y = (\mathit{in}_z \cup \mathit{out}_x) \cap F_y$ \label{cond:gen3}
  \end{enumerate}
  We write $\generators{z}{s}$ for the set of pairs in $\shapes{x} \times
  \shapes{y}$ that generate $s \in \shapes{z}$.
\end{definition}
\begin{lemma}\label{lemgenerate}
\begin{sloppypar}
  Let $x, y, z \in V(T)$ such that $x$ and $y$ are the children of $z$,
  and let $\tau_x \in 2^{\var_x}$ be of shape $(\shapeout{x},
  \shapein{x}) \in \shapes{x}$ and $\tau_y \in 2^{\var_y}$ be of shape
  $(\shapeout{y}, \shapein{y}) \in \shapes{y}$. If $(\shapeout{x},
  \shapein{x})$ and $(\shapeout{y}, \shapein{y})$ generate the shape
  $(\shapeout{z}, \shapein{z}) \in \shapes{z}$, then $\tau = \tau_x
  \cup \tau_y$ is of shape $(\shapeout{z}, \shapein{z})$. Moreover, if
  $(\shapeout{z}, \shapein{z})$ is proper then
  $(\shapeout{x},\shapein{x})$ and $(\shapeout{y},\shapein{y})$ are
  proper.
\end{sloppypar}
\end{lemma}
\begin{proof}
  Suppose $(\shapeout{x}, \shapein{x})$ and $(\shapeout{y},
  \shapein{y})$ generate $(\shapeout{z}, \shapein{z})$. To see that
  $\tau$ satisfies condition~\ref{cond:shape1}, note that a clause is
  satisfied by $\tau$ if and only if it is satisfied by $\tau_x$ or
  $\tau_y$, so $\overline{F_z}(\tau_z) = \overline{F_z}(\tau_x) \cup
  \overline{F_z}(\tau_y) = (\shapeout{x} \cap \overline{F_z}) \cup
  (\shapeout{y} \cap \overline{F_z}) = \shapeout{z}$. For
  condition~\ref{cond:shape2}, let $C \in F_z = F_x \cup F_y$. Without
  loss of generality assume that $C \in F_x$. Suppose $\tau$ does not
  satisfy $C$. Then $\tau_x$ does not satisfy $C$, so we must have
  $C \in \shapein{x}$ because $\tau_x$ is of shape $(\shapeout{x},
  \shapein{x})$. But $\tau_y$ does not satisfy $C$ either, so $C
  \notin \shapeout{y}$. Combining these statements, we get $C \in
  \shapein{x} \setminus \shapeout{y}$. Because $(\shapeout{x},
  \shapein{x})$ and $(\shapeout{y}, \shapein{y})$ generate
  $(\shapeout{z}, \shapein{z})$ we have $\shapein{x} = (\shapein{z}
  \cup \shapeout{y}) \cap F_x$ by condition~\ref{cond:gen2}. It
  follows that $C \in \shapein{z}$.

  The assignments $\tau_x$ and $\tau_y$ are of shapes $(\shapeout{x},
  \shapein{x})$ and $(\shapeout{y}, \shapein{y})$ so $\shapeout{x} \in \projset{\overline{F_x}}{\var_x}$ and $\shapeout{y} \in
  \projset{\overline{F_y}}{\var_y}$ by condition~\ref{cond:shape1}.
  Suppose $(\shapeout{z}, \shapein{z})$ is proper. Then there is an
  assignment $\rho \in 2^{\overline{\var_z}}$ such that $\shapein{z}= F_z(\rho)$. The shapes $(\shapeout{x},\shapein{x})$ and
  $(\shapeout{y}, \shapein{y})$ generate $(\shapeout{z},
  \shapein{z})$, so $\shapein{x} = (\shapein{z} \cup \shapeout{y})
  \cap F_x$. Thus $\shapein{x} = (F_z(\rho) \cup
  \overline{F_y}(\tau_y)) \cap F_x$. Equivalently, $\shapein{x} =
  (F_z(\rho) \cap F_x) \cup (\overline{F_y}(\tau_y) \cap F_x)$. Since
  $F_x \subseteq F_z$ and $F_x \subseteq \overline{F_y}$ this can be
  rewritten once more as $\shapein{x} = F_x(\rho) \cup
  F_x(\tau_y)$. The domains $\overline{\var_z}$ of $\rho$ and $\var_y$
  of $\tau_y$ are disjoint, so $F_x(\rho) \cup F_x(\tau_y) = F_x(\rho
  \cup \tau_y)$. Because $\overline{\var_z} \cup \var_y =
  \overline{\var_x}$ it follows that $\shapein{x} \in
  \projset{F_x}{\overline{\var_x}}$ and so $(\shapeout{x},
  \shapein{x})$ is proper. A symmetric argument shows that
  $(\shapeout{y}, \shapein{y})$ is proper.
\end{proof}
\begin{corollary}\label{corpropergen}
  Let $x, y, z \in V(T)$ such that $x$ and $y$ are the children of $z$ in~$T$,
  and let $s \in \shapes{z}$ be proper. Suppose $s_x \in \shapes{x}$ and $s_y
  \in \shapes{y}$ generate $s$ and both $N_x(s_x)$ and $N_y(s_y)$ are
  nonempty. Then $s_x$ and $s_y$ are proper.
\end{corollary}
\begin{lemma}\label{lemshape}
\begin{sloppypar}
  Let $x, y, z \in V(T)$ such that $x$ and $y$ are the children of $z$, and
  let $\tau \in 2^{\var_z}$ be a truth assignment of shape $(\shapeout{z},
  \shapein{z}) \in \shapes{z}$. Let $\tau_x~=~\proj{\tau}{\var_x}$ and $\tau_y
  = \proj{\tau}{\var_y}$. There are unique shapes $(\shapeout{x}, \shapein{x})
  \in \shapes{x}$ and $(\shapeout{y}, \shapein{y}) \in \shapes{y}$ generating
  $(\shapeout{z}, \shapein{z})$ such that $\tau_x$ has shape $(\shapeout{x},
  \shapein{x})$ and $\tau_y$ has shape $(\shapeout{y}, \shapein{y})$.
\end{sloppypar}
\end{lemma}
\begin{proof}
  We define $\shapeout{x} = \overline{F_x}(\tau_x)$, $\shapeout{y} =
  \overline{F_y}(\tau_y)$ and let $\shapein{x} = (\shapein{z} \cap
  F_x) \cup F_x(\tau_y)$, $\shapein{y} = (\shapein{z} \cap F_y) \cup
  F_y(\tau_x)$. We prove that $(\shapeout{x}, \shapein{x})$ and
  $(\shapeout{y}, \shapein{y})$ generate $(\shapeout{z},
  \shapein{z})$. Since $\tau$ has shape $(\shapeout{z}, \shapein{z})$
  by condition~\ref{cond:shape1} we have $\shapeout{z} =
  \overline{F_z}(\tau)$. We further have $\overline{F_z}(\tau) =
  \overline{F_z}(\tau_x) \cup \overline{F_z}(\tau_y)$ by choice of
  $\tau_x$ and $\tau_y$. Because $\overline{F_z} \subseteq
  \overline{F_x}$ and $\overline{F_z} \subseteq \overline{F_y}$ we get
  $\overline{F_z}(\tau) = (\overline{F_x}(\tau_x) \cap \overline{F_z})
  \cup (\overline{F_y}(\tau_y) \cap \overline{F_z})$ and thus
  $\overline{F_z}(\tau) = (\shapeout{x} \cup \shapeout{y}) \cap
  \overline{F_z}$. That is, condition~\ref{cond:gen1} is satisfied.
  From $F_x \subseteq \overline{F_y}$ and $F_y \subseteq
  \overline{F_x}$ it follows that $F_x(\tau_y) =
  \overline{F_y}(\tau_y) \cap F_x$ and $F_y(\tau_x) =
  \overline{F_x}(\tau_x) \cap F_y$. Thus $F_x(\tau_y) = \shapeout{y}
  \cap F_x$ and $F_y(\tau_x) = \shapeout{x} \cap F_y$ by construction
  of $\shapeout{x}$ and $\shapeout{y}$. By inserting in the
  definitions of $\shapein{x}$ and $\shapein{y}$ we get $\shapein{x} =
  (\shapein{z} \cap F_x) \cup (\shapeout{y} \cap F_x)$ and
  $\shapein{y} = (\shapein{z} \cap F_y) \cup (\shapeout{x} \cap F_y)$,
  so conditions~\ref{cond:gen2} and \ref{cond:gen3} are satisfied. We
  conclude that $(\shapeout{x}, \shapein{x})$ and $(\shapeout{y},
  \shapein{y})$ generate $(\shapeout{z}, \shapein{z})$.
  
  We proceed to showing that $\tau_x$ is of shape $(\shapeout{x},
  \shapein{x})$. Condition~\ref{cond:shape1} is satisfied by
  construction. To see that condition~\ref{cond:shape2} holds, pick
  any $C \in F_x$ not satisfied by $\tau_x$. If $\tau_y$ satisfies
  $C$, then $C \in F_x(\tau_y) \subseteq \shapein{x}$. Otherwise,
  $\tau = \tau_x \cup \tau_y$ does not satisfy $C$. Since $\tau$ of
  shape $(\shapeout{z}, \shapein{z})$ this implies $C \in
  \shapein{z}$. Again we get $C \in \shapein{x}$ as $\shapein{z} \cap
  F_x \subseteq \shapein{x}$. The proof that $\tau_y$ has shape
  $(\shapeout{y}, \shapein{y})$ is symmetric.

  To show uniqueness, let $(\shapeout{x}', \shapein{x}') \in
  \shapes{x}$ and $(\shapeout{y}', \shapein{y}') \in \shapes{y}$
  generate $(\shapeout{z}, \shapein{z})$, and suppose $\tau_x$ has
  shape $(\shapeout{x}', \shapein{x}')$ and $\tau_y$ has shape
  $(\shapeout{y}', \shapein{y}')$. From condition~\ref{cond:shape1} we
  immediately get $\shapeout{x}' = \overline{F_x}(\tau_x) =
  \shapeout{x}$ and $\shapeout{y}' = \overline{F_y}(\tau_y) =
  \shapeout{y}$. Since the pairs $(\shapeout{x}',
  \shapein{x}')$,$(\shapeout{y}', \shapein{y}')$ and $(\shapeout{x},
  \shapein{x})$, $(\shapeout{y},\shapein{y})$ both generate
  $(\shapeout{z}, \shapein{z})$, it follows from
  condition~\ref{cond:gen2} that $\shapein{x}' = \shapein{x}$ and
  $\shapein{y}' = \shapein{y}$.
\end{proof}
\begin{lemma}\label{lem:counting} Let $x, y, z \in V(T)$ such that $x$ and $y$ are the children
  of $z$ in $T$, and let $s \in \shapes{z}$. The following equality holds.
\begin{align}
  n_z(s) = \sum_{(s_x, s_y) \in \generators{z}{s}} n_x(s_x)\: n_y(s_y)\label{equ:lemcounting}
\end{align}
\end{lemma}
\longversion{\begin{proof}
  Let $M(s) = \bigcup_{(s_x, s_y) \in \generators{z}{s}} N_x(s_x)
  \times N_y(s_y)$. We first show that the function $f: \tau \mapsto
  (\tau|_{\var_x}, \tau|_{\var_y})$ is a bijection from $N_z(s)$ to
  $M(s)$. By Lemma~\ref{lemshape} for every $\tau \in N_z(s)$ there is
  a pair $(s_x, s_y) \in \generators{z}{s}$ such that $\tau|_{\var_x}
  \in N_x(s_x)$ and $\tau|_{\var_y} \in N_y(s_y)$. So $f$ is
  \emph{into}. By Lemma~\ref{lemgenerate}, for every pair of
  assignments $\tau_x \in N_x(s_x), \tau_y \in N_y(s_y)$ with
  $(s_x,s_y) \in \generators{z}{s}$ the assignment $\tau_x \cup
  \tau_y$ is in $N_z(s)$. Hence $f$ is \emph{surjective}. It is easy
  to see that $f$ is \emph{injective}, so $f$ is indeed a bijection.
  
  We prove that $\Card{M(s)}$ is equivalent to the right hand side of
  Equality~\ref{equ:lemcounting}. Since $\Card{N_x(s_x) \times
    N_y(s_y)} = n_x(s_x)\:n_y(s_y)$ for every pair $(s_x, s_y) \in
  \generators{z}{s}$, we only have to show that the sets $N_x(s_x)
  \times N_y(s_y)$ and $N_x(s_x') \times N_y(s_y')$ are disjoint for
  distinct pairs $(s_x, s_y), (s_x', s_y') \in \generators{z}{s}$. Let
  $(s_x, s_y), (s_x', s_y') \in \generators{z}{s}$ and suppose
  $(N_x(s_x) \times N_y(s_y)) \cap (N_x(s_x') \times N_y(s_y'))$ is
  nonempty. Pick any $(\tau_x, \tau_y) \in (N_x(s_x) \times N_y(s_y))
  \cap (N_x(s_x') \times N_y(s_y'))$. The function $f$ is a bijection,
  so $\tau_x \cup \tau_y \in N_z(s)$. By Lemma~\ref{lemshape} there is
  at most one pair $(s_x'', s_y'') \in \generators{z}{s}$ of shapes
  such that $\tau_x \in N_x(s_x'')$ and $\tau_y \in N_y(s_y'')$, so
  $(s_x, s_y) = (s_x'', s_y'') = (s_x', s_y')$.
\end{proof}}
\begin{corollary}\label{corcountingproper}
  Let $x, y, z \in V(T)$ such that $x$ and $y$ are the children of $z$ in
  $T$, and let $s \in \shapes{z}$ be proper. Let $P = \SB (s_x, s_y)
  \in \generators{z}{s} \SM s_x$ and $s_y$ are proper$\SE$. The
  following equality holds.
  \begin{align}
  n_z(s) = \sum_{\substack{(s_x, s_y) \in P}} n_x(s_x)\: n_y(s_y)\label{equ:lemcountingproper}
\end{align}
\end{corollary}
\begin{proof}
  By Corollary~\ref{corpropergen} the product $n_x(s_x) n_y(s_y)$ is
  nonzero only if $s_x$ and $s_y$ are proper, for any pair $(s_x, s_y)
  \in \generators{z}{s}$. In combination with (\ref{equ:lemcounting})
  this implies (\ref{equ:lemcountingproper}).
\end{proof}
Corollary~\ref{corcountingproper} in combination with
Lemma~\ref{lememptyproper} implies that, for each $z \in V(T)$, it
is enough to compute the values $n_z(s)$ for \emph{proper} shapes $s
\in \shapes{z}$. To turn this insight into a polynomial time dynamic
programming algorithm, we still have to show that the number of proper
shapes in $\shapes{z}$ can be polynomially bounded, and that the set
of such shapes can be computed in polynomial time. We will achieve
this by specifying a subset of $\shapes{z}$ for each $z \in V(T)$ that
contains all proper shapes and can be computed in polynomial time.

We define families $\tuplestop{z}$ and $\tuplesbot{z}$
of sets of variables for each node $z \in V(T)$, as follows.
\begin{align*}
  \tuplestop{z} = \SB X \subseteq \var_z \SM \exists C \in
  \overline{F_z}\mtext{ such that } X = \var_z \cap \var(C) \SE \\
  \tuplesbot{z} = \SB X \subseteq \overline{\var_z} \SM \exists C \in F_z
  \mtext{ such that }
  X = \overline{\var_z} \cap \var(C) \SE
\end{align*}
The next lemma follows from the definition of a decomposition tree's index.
\begin{lemma}\label{lemxsbounded}
  For every node $z \in V(T)$, $\max(\Card{\tuplestop{z}}, \Card{\tuplesbot{z}}) \leq k$.
\end{lemma}
Let $z \in V(T)$ and let $f$ be a function with domain $\tuplestop{z}$ that maps
every set $X$ to some projection $f(X) \in
\projset{\proj{\overline{F_z}}{X}}{X}$. We denote the set of such functions by
$\outfunctions{z}$. Symmetrically, we let $\infunctions{z}$ denote the
set of functions~$g$ that map every set $Y \in \tuplesbot{z}$ to some
projection $g(Y) \in \projset{\proj{F_z}{Y}}{Y}$.
\begin{lemma}\label{lemprojbounded}
  For every $z \in V(T)$, $\Card{\outfunctions{z}} \leq (m+1)^{k}$ as well as
  $\Card{\infunctions{z}} \leq (m+1)^k$.
\end{lemma}
\begin{proof}
  By Proposition~\ref{propprojbounded} that the cardinality of
  $\projset{\proj{\overline{F_z}}{X}}{X}$ is bounded by $m + 1$ for every $X
  \in \tuplestop{z}$. In combination with Lemma~\ref{lemxsbounded} this yields
  $\Card{\outfunctions{z}} \leq (m+1)^{k}$. The proof of
  $\Card{\infunctions{z}} \leq (m+1)^{k}$ is symmetric.
\end{proof}
Let $\mathsf{union}(f)$ denote $\bigcup_{X \in \mathsf{dom}(f)}f(X)$,
where $\mathsf{dom}(f)$ is the domain of $f$. We define the set of
\emph{restricted} shapes for $z \in V(T)$ as follows.
\begin{align*}
  \rshapes{z} = \SB (\mathit{out}, \mathit{in}) \in \shapes{z} \SM &\exists f \in
  \outfunctions{z} \mtext{ s.t. } \mathit{out} = \mathsf{union}(f) \\
  \wedge  &\exists g \in \infunctions{z} \mtext{ s.t. } \mathit{in} = \mathsf{union}(g) \SE
\end{align*}
Every pair $(f, g) \in \outfunctions{z} \times \infunctions{z}$ uniquely
determines a shape in $\rshapes{z}$. Accordingly, Lemma~\ref{lemprojbounded}
allows us to bound the cardinality of $\rshapes{z}$ as follows.
\begin{corollary}\label{correstrictedbounded}
  For any $z \in V(T)$, $\Card{\rshapes{z}} \leq (m+1)^{2k}$.
\end{corollary}
\begin{lemma}\label{lemrestrictedproper}
  Let $z \in V(T)$ and let $s \in \shapes{z}$ be proper. Then $s \in
  \rshapes{z}$.
\end{lemma}
\newcommand{\pflemrestrictedproper}[0]{
\begin{proof}
  Let $s = (\shapeout{},\shapein{})$. We show that there are functions $f \in
  \outfunctions{z}$ and $g \in \infunctions{z}$ such that $\shapeout{} =
  \mathsf{union}(f)$ and $\shapein{} = \mathsf{union}(g)$. Because $s$ is
  proper we have $\shapeout{} \in \projset{\overline{F_z}}{\var_z}$ and
  $\shapein{} \in \projset{F_z}{\overline{\var_z}}$, so there must be truth
  assignments $\sigma \in 2^{\var_z}$ and $\tau \in 2^{\overline{\var_z}}$
  such that $\shapeout{} = \overline{F_z}(\sigma)$ and $\shapein{} =
  F_z(\tau)$. We define $f$ as follows. For each $X \in \tuplestop{z}$ we let
  $f(X) = \proj{\overline{F_z}}{X}(\proj{\sigma}{X})$. The assignment $\sigma$
  is defined on $X \subseteq \var_z$, so $\proj{\sigma}{X} \in 2^X$ and $f(X)
  \in \projset{\proj{\overline{F_z}}{X}}{X}$. That is, $f \in
  \outfunctions{z}$. Symmetrically, we let $g(X) =
  \proj{F_z}{X}(\proj{\tau}{X})$ for each $X \in \tuplesbot{z}$. Since $\tau$
  is defined on $X \subseteq \overline{\var_z}$ we have $\proj{\tau}{X} \in
  2^{X}$ and $g(X) \in \projset{\proj{F_z}{X}}{X}$, so $g \in \infunctions{z}$.

  Pick an arbitrary $C \in \overline{F_z}$ and let $X = \var(C) \cap
  \var_z$. We show that $C \in \shapeout{}$ if and only if $C \in
  \mathsf{union}(f)$. Suppose $C \in \shapeout{} =
  \overline{F_z}(\sigma)$. The assignment $\sigma$ has domain $var_z$,
  so $\proj{\sigma}{X}$ satisfies $C$ because $\sigma$ does. That is,
  $C \in \overline{F_z}(\proj{\sigma}{X})$. By choice of $X$ we have
  $C \in \proj{F_z}{X}$, so $C \in \overline{F_z}(\proj{\sigma}{X})
  \cap \proj{F_z}{X}$. Since $\proj{F_z}{X} \subseteq F_z$ we get
  $\overline{F_z}(\proj{\sigma}{X}) \cap \proj{F_z}{X} =
  \proj{\overline{F_z}}{X}(\proj{\sigma}{X})$. So $C \in
  \proj{\overline{F_z}}{X}(\proj{\sigma}{X}) = f(X)$ and thus $C \in
  \mathsf{union}(f)$. For the converse direction, suppose $C \in
  \mathsf{union}(f)$. That is, $C \in f(Y) =
  \proj{\overline{F_z}}{Y}(\proj{\sigma}{Y})$ for some $Y \in
  \tuplestop{z}$. Then in particular $C \in \overline{F_z}(\sigma) =
  \shapeout{}$. We conclude that $\mathsf{union}(f) =
  \shapeout{}$. The proof of $\mathsf{union}(g) = \shapein{}$ is
  symmetric.
\end{proof}}
\longversion{\pflemrestrictedproper}
This shows that if we can determine the values $n_z(s)$ for every $z \in V(T)$
and $s \in \rshapes{z}$, we can determine the values $n_z(s')$ for every
proper shape $s' \in~\shapes{z}$. More specifically, as long as we can
determine lower bounds for $n_z(s)$ for every $s \in \rshapes{z}$ and the
exact values of $n_z(s)$ for proper $s$, we can compute the correct values for
all proper shapes for every tree node.
\begin{definition}
  For $z \in V(T)$, a \emph{lower bounding function} (for $z$) associates with
  each $s \in \rshapes{z}$ a value $l_z(s)$ such that $l_z(s) \leq n_z(s)$ and
  $l_z(s) = n_z(s)$ if $s$ is proper.
\end{definition}
Let $x,y,z \in V(T)$ such that $x$ and $y$ are the children of $z$. For each $s
\in \shapes{z}$ we write $\rgenerators{z}{s} = \generators{z}{s} \cap
(\rshapes{x} \times \rshapes{y})$.
\begin{lemma}\label{lemdpequality}
  Let $x,y,z \in V(T)$ such that $x$ and $y$ are the children of $z$. Let $l_x$
  and $l_y$ be lower bounding functions for $x$ and $y$. Let $l_z$ be the
  function defined as follows. For each $s \in \rshapes{z}$, we let
  \begin{align}
     l_z(s) = \sum_{\substack{(s_x, s_y) \in \rgenerators{z}{s}}} l_x(s_x)\:
     l_y(s_y).
   \end{align}
   Then $l_z$ is a lower bounding function for $z$.
 \end{lemma}
\newcommand{\pflemdpequality}[0]{
\begin{proof}
  The inequality $l_z(s) \leq n_z(s)$ follows from $\rshapes{x} \subseteq
  \shapes{x}$ and $\rshapes{y} \subseteq \shapes{y}$, in combination with
  equality~(\ref{equ:lemcounting}) and the fact that $l_x$ and $l_y$ are lower
  bounding functions for $x$ and $y$. By Lemma~\ref{lemrestrictedproper} the
  set $\rgenerators{z}{s}$ contains all pairs $(s_x,s_y) \in
  \generators{z}{s}$ such that $s_x$ and $s_y$ are proper. It follows from
  Corollary~\ref{corcountingproper} and $l_x(s_x) = n_x(s_x)$, $l_y(s_y) =
  n_y(s_y)$ for proper $s_x,s_y$ that $l_z (s) \geq n_z(s)$ and thus $l_z(s) =
  n_z(s)$ for proper $s$. We conclude that $l_z$ is a lower bounding function
  for $z$.
\end{proof}}
\longversion{\pflemdpequality}

\longversion{
\begin{lemma}\label{lemcomputeshapes}
  There is a polynomial $p$ such that for any $z \in V(T)$, the set
  $\rshapes{z}$ can be computed in time $m^{2k} p(l)$, where $l$ is
  the length of $F$.
\end{lemma}
\newcommand{\pflemcomputeshapes}[0]{
\begin{proof}
  To compute $\rshapes{z}$, we compute all pairs
  $(\mathsf{union}(f),\mathsf{union}(g))$ for $(f,g) \in \outfunctions{z}
  \times \infunctions{z}$. To compute the set $\tuplestop{z}$, we run through
  all clauses $C \in \overline{F_z}$ and determine $\var(C) \cap \var_z$. This
  can be done in time polynomial in $l$, and the same holds for the set
  $\tuplesbot{z}$. A function $f \in \outfunctions{z}$ maps each $X \in
  \tuplestop{z}$ to a set $f(X) \in \projset{\proj{\overline{F_z}}{X}}{X}$. By
  Proposition~\ref{propprojbounded} the set
  $\projset{\proj{\overline{F_z}}{X}}{X}$ can by computed in time polynomial
  in $l$ for each $X \in \tuplestop{z}$. Going through all possible pairs
  $(f,g) \in \outfunctions{z} \times \infunctions{z}$ amounts to going through
  all possible combinations of choices of $f(X) \in
  \projset{\proj{\overline{F_z}}{X}}{X}$ for each $X \in \tuplestop{z}$ and
  $g(X') \in \projset{\proj{F_z}{X'}}{X'}$ for each $X' \in \tuplesbot{z}$, of
  which there are at most $(m+1)^{2k}$. For each such pair $(f,g)$ we compute
  the sets $\mathsf{union}(f)$ and $\mathsf{union}(g)$, which can be done in
  time polynomial in $l$.
\end{proof}}}
\longversion{\pflemcomputeshapes}
\longversion{
\begin{lemma}\label{lemgeneratecheck}
  Let $x,y,z \in V(T)$ such that $x$ and $y$ are the children of $z$. Let
  $s_x \in \shapes{x}$, $s_y \in \shapes{y}$, and $s_z \in
  \shapes{z}$. It can be decided in time $O(l^2)$ whether $s_x$ and
  $s_y$ generate $s_z$, where $l$ is the length of $F$.
\end{lemma}
\begin{proof}
  We only have to check conditions~\ref{cond:gen1} to \ref{cond:gen3}, which
  can easily be done in time quadratic in $l$ since the sets of clauses
  involved have length at most $l$.
\end{proof}}
\longversion{
\begin{lemma}\label{lemleaf}
  For any leaf node $z \in V(T)$ a lower bounding function for $z$ can be
  computed in time $O(l)$, where $l$ is the length of $F$.
\end{lemma}
\newcommand{\pflemleaf}[0]{
\begin{proof}
  Every leaf $z \in V(T)$ is either associated with a clause $C \in F$ or a
  variable $v \in \var(F)$. In the first case, $\var_z = \emptyset$ and so
  $\tuplestop{z} = \emptyset$ if $\overline{F_z} = \emptyset$ or
  $\tuplestop{z} = \{\emptyset\}$. It follows that the set $\outfunctions{z}$
  only contains the empty function or the function $f$ with domain
  $\{\emptyset\}$ such that $f(\emptyset) = \emptyset$. For the set
  $\tuplesbot{z}$ we get $\tuplesbot{z} = \{\var(C)\}$ for the unique clause
  $C \in F_z$. Since $\proj{F_z}{\var(C)} = \{C\}$ we have
  $\projset{\proj{F_z}{\var(C)}}{\var(C)} = \{ \{C\}, \emptyset\}$ and thus
  $\infunctions{z} = \{g, g'\}$, where $g$ is the function with domain
  $\{\var(C)\}$ such that $g(\var(C)) = \{C\}$ and $g'$ is the function with
  domain $\{\var(C)\}$ such that $g'(\var(C)) = \emptyset$. It follows that
  $\rshapes{z}$ only contains the shapes $(\emptyset, \emptyset)$ and
  $(\emptyset, \{C\})$. The set $\var_z$ is empty, so $2^{\var_z}$ contains
  only the empty assignment which does not satisfy any clause. Hence
  $n_z((\emptyset, \emptyset)) = 0$ and $n_z((\emptyset, \{C\})) = 1$.

  In the second case, $\var_z = \{v\}$ for some variable $v \in
  \var(F)$. Since $F_z = \emptyset$ we have $\tuplesbot{z} = \emptyset$ and so
  $\infunctions{z}$ only contains the empty function. The set $\tuplestop{z}$
  contains $\{v\}$, and the empty set if there is a clause $C \in F$ with $v
  \notin \var(C)$. We get $\projset{\proj{\overline{F_z}}{\{v\}}}{\{v\}} =
  \{F^{+}_v, F^{-}_v\}$, where $F^+_v$ is the set of clauses of $F$ with a
  positive occurrence of $v$, and $F^-_v$ is the set of clauses $F$ with a
  negative occurrence of $v$. Moreover,
  $\projset{\proj{\overline{F_z}}{\emptyset}}{\emptyset} = \{\emptyset\}$. It
  follows that $\rshapes{z} = \{ ( F^+_v, \emptyset), (F^-_v, \emptyset) \}$.
  The set $2^{\var_z}$ only contains the assignments $\tau_0$ with $\tau_0(v)
  = 0$ and $\tau_1$ with $\tau_1(v) = 1$, and $\overline{F_z}(\tau_0) = F^-_v$
  and $\overline{F_z}(\tau_1) = F^+_v$. This implies $n_z((F^+_v, \emptyset))
  = 1$ and $n_z((F^-_v,\emptyset)) = 1$.

  In either case the set $\rshapes{z}$ and the values $n_z(s)$ for each $s \in
  \rshapes{z}$ can be computed in time $O(l)$. These values trivially provide
  a lower bounding function for $z$.
\end{proof}}
\pflemleaf}

\longversion{
\begin{lemma}\label{leminner}
  There is a polynomial $p$ such that for any inner node $z \in V(T)$,
  a lower bounding function for $z$ can be computed in time $m^{6k}
  p(l)$, provided that lower bounding functions have already been computed for
  both children of $z$, where $l$ denotes the length of $F$.
\end{lemma}
\newcommand{\pfleminner}[0]{
\begin{proof}
  By Lemma~\ref{lemcomputeshapes}, there is a polynomial $q$
  (independent of $z$) such that the set $\rshapes{z}$ can be computed
  in time $O(m^{2k} q(l))$. Let $x$ and $y$ denote the children
  of $z$, and let $l_x$ and $l_y$ be lower bounding functions for $x$
  and $y$. We compute a lower bounding function $l_z$ for $z$ as
  follows. Initially, we set $l_z(s_z) = 0$ for all $s_z \in
  \rshapes{z}$. We then run through all triples of shapes $s_x \in
  \rshapes{x}$, $s_y \in \rshapes{y}$, and $s_z \in \rshapes{z}$ and
  check whether $s_x$ and $s_y$ generate $s_z$. If that is the case,
  we add $l_x(s_x)\: l_y(s_y)$ to $l_z(s_z)$.

  Correctness follows from Lemma~\ref{lemdpequality} and the fact that
  $l_x, l_y$ are lower bounding functions for $x$ and $y$. The bound
  on the runtime is obtained as follows. By
  Corollary~\ref{correstrictedbounded} there are at most
  $(m+1)^{6k}$ triples $(s_x, s_y, s_z)$
  of shapes that have to be considered. For each one, one can decide
  whether $s_x$ and $s_y$ generate $s_z$ in time $O(l^2)$ by
  Lemma~\ref{lemgeneratecheck}. Depending on the outcome of that
  decision we may have to multiply two integers $l_x(s_x)$ and
  $l_y(s_y)$, adding the result to $l_z(s_z)$. These values are
  bounded from above by $2^{\Card{\var(F)}} \leq 2^l$, so their binary
  representations have size $O(l)$ and these arithmetic operations can
  be carried out in time polynomial in $l$.
\end{proof}}
\pfleminner}

\longversion{
\begin{lemma}\label{lemfullalgorithm}
  There is a polynomial $p$ such that a lower bounding function for $z$
  can be computed for every $z \in V(T)$ in time $m^{6k} p(l)$, where
  $l$ is the length of $F$.
\end{lemma}
\newcommand{\pflemfullalgorithm}[0]{
\begin{proof}
  By Lemma~\ref{lemleaf}, a lower bounding function for a leaf of $T$ can be
  computed in time $O(l)$. The number of leaves of $T$ is in $O(l)$, so we can
  compute lower bounding functions for all of them in time $O(l^2)$. By
  Lemma~\ref{leminner}, we can then compute lower bounding functions for each
  inner node $z \in V(T)$ in a bottom up manner. For each inner node $z$, a lower
  bounding function can computed in time $m^{6k} q(l)$ by
  Lemma~\ref{leminner}, where $q$ is a polynomial independent of $z$. The
  number of inner nodes of $T$ is in $O(l)$, so this requires $O(m^{6k} l\:
  q(l))$ time in total.~
\end{proof}}
\pflemfullalgorithm}
\begin{proposition}\label{propcountingwithdecomp}
  There is a polynomial $p$ and an algorithm $\mathbb{A}$ such that
  $\mathbb{A}$, given a CNF formula $F$ and a decomposition tree $(T,
  \delta)$ of $I(F)$, computes the number of satisfying total truth
  assignments of $F$ in time $m^{6k} p(l)$. Here, $m$ denotes
  the number of clauses of $F$, $l$ denotes the length of $F$, and $k
  = \mathit{index}(T, \delta)$.
\end{proposition}
\newcommand{\pfpropcountingwithdecomp}[0]{
\begin{proof}
  By Lemma~\ref{lemfullalgorithm} a lower bounding function $l_r$ for
  the root $r$ of $T$ can be computed in time $m^{6k}\: q(l)$,
  where $q$ is a polynomial independent of $F$. By
  Lemma~\ref{lememptyproper}, the value $n_r((\emptyset, \emptyset))$
  corresponds to the number of satisfying total truth assignments of
  $F$, and the shape $(\emptyset, \emptyset)$ is proper. Since $l_r$
  is a lower bounding function for $r$ it follows that
  $l_r((\emptyset, \emptyset)) = n_r((\emptyset, \emptyset))$.
\end{proof}}
\longversion{\pfpropcountingwithdecomp}
\shortversion{
  \begin{proof}[Sketch]
    We compute lower bounding functions for every node of $T$. It follows from
    Corollary~\ref{correstrictedbounded} that each such function can be
    represented in polynomial space for fixed $k$. Computing lower bounding
    functions for leaf nodes is straightforward. For an inner node $z$ with
    children $x$ and $y$, we proceed as follows. Assume that lower bounding
    functions $l_x$ and $l_y$ for $x$ and $y$ have already been computed. We
    first compute the set $\rshapes{z}$ and set $l_z(s_z) := 0$ for each $s_z
    \in \rshapes{z}$. We then run through all triples $(s_x, s_y, s_z)$ with
    $s_x \in \rshapes{x}, s_y \in \rshapes{y}, s_z \in \rshapes{z}$ and check
    whether $s_x$ and~$s_y$ generate $s_z$ (each check can be done in
    polynomial time). If that is the case, we set $l_z(s_z) := l_z(s_z) +
    l_x(s_x) l_y(s_y)$. By Lemma~\ref{lemdpequality}, the resulting $l_z$ will
    be a lower bounding function. There are at most $(m+1)^{6k}$ such triples
    for each inner node, so this can be done in time $m^{6k} p(l)$ for all
    nodes of $T$, where $p$ is a suitable polynomial independent of
    $F$. Having computed a lower bounding function $l_r$ for the root $r$ of
    $T$, we output $l_r((\emptyset, \emptyset))$, which corresponds to the
    number of satisfying total truth assignments of $F$.
  \end{proof}}
\begin{proof}[of Theorem~\ref{thmmain}]
  Let $\CCC$ be a graph class of bounded symmetric clique\hy width and $F$ a
  CNF formula of length $l$ with $m$ clauses such that $I(F) \in \CCC$. Let
  $k$ be an upper bound for the symmetric clique\hy width of any graph in
  $\CCC$. We compute a decomposition tree $(T, \delta)$ of $I(F)$ such that
  $\mathit{rankw}(T,\delta) = \mathit{rankw}(I(F))$ as follows. Initially, we
  set $k' := 1$. We then repeatedly run the algorithm of
  Theorem~\ref{thmrankdecomp} and increment $k'$ by one until we find a
  decomposition of rank\hy width~$k'$. This will be the case after at most $k$
  steps since $\mathit{rankw}(I(F)) \leq \mathit{scw}(I(F))$ by
  Corollary~\ref{corscwrankw}.  Since $\CCC$ is fixed, we can consider $k$
  (and every $k' \leq k$) a constant, so $(T,\delta)$ can be obtained in time
  $O(\Card{V(I(F))}^3)$ by Theorem~\ref{thmrankdecomp}. Because $2l$ is an upper bound on
  the number of vertices of $I(F)$, this is in $l^{O(1)}$ (assuming that $l
  \geq 2$). By Lemma~\ref{lemscwrankw}, $\mathit{index}(T, \delta) \leq
  2^{\mathit{rankw}(I(F))}$ and thus $\mathit{index}(T, \delta) \leq
  2^{\mathit{scw}(I(F))} \leq 2^k$. By
  Proposition~\ref{propcountingwithdecomp}, the number of satisfying total
  truth assignments of $F$ can be computed in time $m^{6\:\mathit{index}(T,
    \delta)} p(l)$ for some polynomial $p$ independent of $F$, that is, in
  time $m^{O(2^k)} p(l)$. Since $k$ is a constant, this is in $l^{O(1)}$, as
  is the total runtime.
\end{proof}
\section{Conclusion}
\begin{sloppypar}
  We have shown that \#SAT is polynomial\hy time tractable for classes of
  formulas with incidence graphs of bounded symmetric clique\hy width (or
  bounded clique\hy width, or bounded rank\hy width). It would be interesting
  to know whether this problem is tractable under even weaker structural
  restrictions. For instance, it is currently open whether \#SAT is
  polynomial\hy time tractable for classes of formulas of bounded $\beta$\hy
  hypertree width~\cite{GottlobPichler04} (if a corresponding decomposition is
  given).
\end{sloppypar}
\subsubsection*{Acknowledgements} The authors would like to thank an anonymous
referee for suggesting to state the main results in terms of symmetric
clique\hy width instead of Boolean\hy width.
\bibliographystyle{plain} 
\bibliography{literature} 

\begin{thebibliography}{10}

\bibitem{BacchusDalmaoPitassi03}
Fahiem Bacchus, Shannon Dalmao, and Toniann Pitassi.
\newblock Algorithms and complexity results for \#{S}{A}{T} and {B}ayesian
  inference.
\newblock In {\em 44th Annual IEEE Symposium on Foundations of Computer Science
  (FOCS'03)}, pages 340--351, 2003.

\bibitem{BuixuanTelleVatshelle10}
Binh-Minh Bui-Xuan, Jan~Arne Telle, and Martin Vatshelle.
\newblock H-join decomposable graphs and algorithms with runtime single
  exponential in rankwidth.
\newblock {\em Discrete Applied Mathematics}, 158(7):809--819, 2010.

\bibitem{BuixuanTelleVatshelle11}
Binh-Minh Bui-Xuan, Jan~Arne Telle, and Martin Vatshelle.
\newblock Boolean-width of graphs.
\newblock {\em Theoretical Computer Science}, 412(39):5187--5204, 2011.

\bibitem{Courcelle04}
Bruno Courcelle.
\newblock Clique-width of countable graphs: a compactness property.
\newblock {\em Discrete Mathematics}, 276(1-3):127--148, 2004.

\bibitem{HlinenyOum08}
Petr~Hlin\v en{\'y} and Sang il~Oum.
\newblock Finding branch-decompositions and rank-decompositions.
\newblock {\em SIAM J. Comput.}, 38(3):1012--1032, 2008.

\bibitem{FischerMakowskyRavve06}
E.~Fischer, J.~A. Makowsky, and E.~R. Ravve.
\newblock Counting truth assignments of formulas of bounded tree-width or
  clique-width.
\newblock {\em Discr. Appl. Math.}, 156(4):511--529, 2008.

\bibitem{GanianHlineny10}
Robert Ganian and Petr Hlin{\v e}n{\'y}.
\newblock On parse trees and {M}yhill-{N}erode-type tools for handling graphs
  of bounded rank-width.
\newblock {\em Discr. Appl. Math.}, 158(7):851--867, 2010.

\bibitem{GanianHlinenyObdrzalek13}
Robert Ganian, Petr Hlinen{\'y}, and Jan Obdrz{\'a}lek.
\newblock Better algorithms for satisfiability problems for formulas of bounded
  rank-width.
\newblock {\em Fund. Inform.}, 123(1):59--76, 2013.

\bibitem{GaspersSzeider13}
Serge Gaspers and Stefan Szeider.
\newblock Strong backdoors to bounded treewidth {SAT}.
\newblock In {\em Proceedings of FOCS 2013, The 54th Annual Symposium on
  Foundations of Computer Science, Berkeley, California, USA}, to appear.

\bibitem{GottlobPichler04}
Georg Gottlob and Reinhard Pichler.
\newblock Hypergraphs in model checking: acyclicity and hypertree-width versus
  clique-width.
\newblock {\em SIAM J. Comput.}, 33(2):351--378, 2004.

\bibitem{NishimuraRagdeSzeider07}
Naomi Nishimura, Prabhakar Ragde, and Stefan Szeider.
\newblock Solving \#{S}{A}{T} using vertex covers.
\newblock {\em Acta Informatica}, 44(7-8):509--523, 2007.

\bibitem{OrdyniakPaulusmaSzeider13}
Sebastian Ordyniak, Dani{\"e}l Paulusma, and Stefan Szeider.
\newblock Satisfiability of acyclic and almost acyclic {CNF} formulas.
\newblock {\em Theoretical Computer Science}, 481:85--99, 2013.

\bibitem{PaulusmaSlivovskySzeider13}
Dani{\"e}l Paulusma, Friedrich Slivovsky, and Stefan Szeider.
\newblock Model counting for {CNF} formulas of bounded modular treewidth.
\newblock In Natacha Portier and Thomas Wilke, editors, {\em Proceedings of
  STACS 2013}, volume~20 of {\em LIPIcs}, pages 55--66. Leibniz-Zentrum fuer
  Informatik, 2013.

\bibitem{Roth96}
Dan Roth.
\newblock On the hardness of approximate reasoning.
\newblock {\em Artificial Intelligence}, 82(1-2):273--302, 1996.

\bibitem{SamerSzeider10}
Marko Samer and Stefan Szeider.
\newblock Algorithms for propositional model counting.
\newblock {\em J. Discrete Algorithms}, 8(1):50--64, 2010.

\bibitem{SangBeameKautz05}
Tian Sang, Paul Beame, and Henry~A. Kautz.
\newblock Performing {Bayesian} inference by weighted model counting.
\newblock In {\em Proceedings of the 20th national conference on Artificial
  intelligence - Volume 1}, AAAI'05, pages 475--481. AAAI Press, 2005.

\bibitem{Szeider04b}
Stefan Szeider.
\newblock On fixed-parameter tractable parameterizations of {S}{A}{T}.
\newblock In Enrico Giunchiglia and Armando Tacchella, editors, {\em SAT 2003,
  Selected and Revised Papers}, volume 2919 of {\em Lecture Notes in Computer
  Science}, pages 188--202. Springer Verlag, 2004.

\bibitem{Valiant79b}
L.~G. Valiant.
\newblock The complexity of computing the permanent.
\newblock {\em Theoretical Computer Science}, 8(2):189--201, 1979.

\end{thebibliography}

\end{document}